\documentclass[12pt]{article}
\RequirePackage{amssymb,epsf,epic,amsmath}

\def\denseformat{
\setlength{\textheight}{9in}
\setlength{\textwidth}{6.9in}
\setlength{\evensidemargin}{-0.2in}
\setlength{\oddsidemargin}{-0.2in}
\setlength{\headsep}{10pt}
\setlength{\topmargin}{-0.3in}
\setlength{\columnsep}{0.375in}
\setlength{\itemsep}{0pt}
}

\newtheorem{theorem}{Theorem}[section]

\newtheorem{claim}[theorem]{Claim}
\newtheorem{lemma}[theorem]{Lemma}

\newtheorem{fact}[theorem]{Fact}
\newtheorem{comment}[theorem]{Comment}

\def\boldhead#1:{\par\vskip 7pt\noindent{\bf #1:}\hskip 10pt}
\def\ithead#1:{\par\vskip 7pt\noindent{\it #1:}\hskip 10pt}

\def\inline#1:{\par\vskip 7pt\noindent{\bf #1:}\hskip 10pt}
\def\midinline#1:{\par\noindent{\bf #1:}\hskip 10pt}
\def\dnsinline#1:{\par\vskip -7pt\noindent{\bf #1:}\hskip 10pt}
\def\ddnsinline#1:{\newline{\bf #1:}\hskip 10pt}
\def\largeinline#1:{\par\vskip 7pt\noindent{\large\bf #1:}\hskip 10pt}
\long\def\comment #1\commentend{}
\long\def\commhide #1\commhideend{}
\long\def\commfull #1\commend{#1}
\long\def\commabs #1\commenda{}
\long\def\commtim #1\commendt{#1}
\long\def\commb #1\commbend{}
\long\def\commedit #1\commeditend{} 

\long\def\commB #1\commBend{}       

\long\def\commex #1\commexend{}     

\long\def\commsiena #1\commsienaend{}  
                                         
\long\def\commBI #1\commBIend{}  
                                         
\long\def\CProof #1\CQED{}
\def\qed{\mbox{}\hfill $\Box$\\}
\def\blackslug{\hbox{\hskip 1pt \vrule width 4pt height 8pt
    depth 1.5pt \hskip 1pt}}
\def\QED{\quad\blackslug\lower 8.5pt\null\par}

\long\def\PPP#1{\noindent{\bf Proof:}{ #1}{\quad\blackslug\lower 8.5pt\null}}

\long\def\denspar #1\densend
{#1}

\newif\ifnotesw\noteswtrue
   {\ifnotesw\marginpar[\hfill\(\top\)]{\(\top\)}\fi}%
   {\ifnotesw\marginpar[\hfill\(\bot\)]{\(\bot\)}\fi}

\newcommand{\mnote}[1]%
    {\ifnotesw\marginpar%
        [{\scriptsize\it\begin{minipage}[t]{\marginparwidth}
        \raggedleft#1%
                        \end{minipage}}]%
        {\scriptsize\it\begin{minipage}[t]{\marginparwidth}
        \raggedright#1%
                        \end{minipage}}%
    \fi}

\def\tR{{\tilde R}}

\def\MathF{\hbox{\rm I\kern-2pt F}}
\def\MathP{\hbox{\rm I\kern-2pt P}}
\def\MathR{\hbox{\rm I\kern-2pt R}}
\def\MathZ{\hbox{\sf Z\kern-4pt Z}}
\def\MathN{\hbox{\rm I\kern-2pt I\kern-3.1pt N}}
\def\MathC{\hbox{\rm \kern0.7pt\raise0.8pt\hbox{\footnotesize I}
\kern-4.2pt C}}
\def\MathQ{\hbox{\rm I\kern-6pt Q}}

\newsavebox{\ttop}\newsavebox{\bbot}

\def\eps{\epsilon}

\denseformat

\begin{document}

\def\varrhol{{et al.~}}
\def\Sp{\mathit{Sp}}
\def\ReadDg{\mathit{Read\_Edge}}
\def\wmax{{{\hat \omega}}}
\def\ttl{\mathit{ttl}}
\def\Rnd{\mathit{Round}}
\def\deg{\mathit{deg}}
\def\ctr{\mathit{ctr}}
\def\ScDgs{\mathit{Scanned\_Edges}}
\def\SCD{\mathit{SCANNED}}
\def\NOTSCD{\mathit{NOTSCANNED}}
\def\CRASH{\mathit{CRASH}}
\def\CR{\mathit{CRASH}}
\def\CRASHED{\mathit{CRASHED}}
\def\SyncIncr{\mathit{Sync\_Incr}}
\def\actdeg{\mathit{actdeg}}
\def\OLD{\mathit{OLD}}
\def\NEW{\mathit{NEW}}
\def\mark{\mathit{mark}}
\def\AsyncRnd{\mathit{Async\_Rnd}}
\def\status{\mathit{status}}
\def\scstat{\mathit{scan\_status}}
\def\lab{\mathit{label}}
\def\seclabel{\mathit{sec\_label}}
\def\own{\mathit{own}}
\def\SELF{\mathit{SELF}}
\def\PEER{\mathit{PEER}}
\def\DynRnd{\mathit{Dyn\_Rnd}}
\def\Delete{\mathit{Delete}}
\def\XReplace{\mathit{XReplace}}
\def\UpdLab{\mathit{Update\_Label}}
\def\Crash{\mathit{Crash}}
\def\CrashLoop{\mathit{CrashLoop}}
\def\CrashItern{\mathit{CrashItern}}
\def\fetched{\mathit{fetched}}
\def\TRUE{\mathit{TRUE}}
\def\FALSE{\mathit{FALSE}}
\def\crashstat{\mathit{crash\_status}}
\def\wmax{\hat{\omega}}
\def\END{\mathit{END}}
\def\te{\tilde{e}}
\def\tu{\tilde{u}}
\def\che{\check{e}}
\def\chu{\check{u}}
\def\Old{\mathit{Old}}
\def\tR{\tilde{R}}
\def\eps{{\epsilon}}
\def\UpLab{\mathit{Update\_Label}}
\def\chP{\check{P}}
\def\VOID{\mathit{VOID}}

\newcommand {\ignore} [1] {}

\title{The MST of Symmetric Disk Graphs \\(in Arbitrary Metrics) is Light}

\author{
Shay Solomon\thanks{
        Department of Computer Science,
        Ben-Gurion University of the Negev, POB 653, Beer-Sheva 84105, Israel.
         \newline E-mail: {\tt \{shayso\}@cs.bgu.ac.il}
        \newline This research has been supported by the
                 Clore Fellowship grant No.\ 81265410  and
 by the BSF grant No. 2008430.
        \newline Partially supported by the
       Lynn and William Frankel
        Center for Computer Sciences.}}
\date{\empty}

\date{\empty}

\begin{titlepage}
\def\thepage{}
\maketitle

\begin{abstract}
Consider an $n$-point metric $M = (V,\delta)$,
and a transmission range assignment $r: V \rightarrow \mathbb R^+$
that maps each point $v \in V$ to the disk of radius $r(v)$ around it.
The \emph{symmetric disk graph} (henceforth, SDG) that corresponds to $M$ and $r$
is the undirected graph over $V$ whose edge set includes
an edge $(u,v)$ if both $r(u)$ and $r(v)$ are no smaller than $\delta(u,v)$.
SDGs are often used to model wireless communication networks.

Abu-Affash, Aschner, Carmi and Katz (SWAT 2010, \cite{AACK10}) 
showed that for any \emph{2-dimensional Euclidean} $n$-point metric $M$, the weight of the MST of every \emph{connected} SDG for $M$
is $O(\log n) \cdot w(MST(M))$, and that this bound is tight.
However, the upper bound proof of \cite{AACK10} relies heavily on basic geometric properties of 2-dimensional Euclidean metrics,
and does not extend to higher dimensions.
A natural question that arises is whether this surprising upper bound of 
\cite{AACK10} can be generalized for wider families of metrics, such as 3-dimensional Euclidean metrics.

In this paper we generalize the upper bound of Abu-Affash et al.\ \cite{AACK10} for Euclidean metrics
of any dimension. Furthermore, our upper bound 
extends to \emph{arbitrary metrics} and, in particular, it applies to any of the normed spaces $\ell_p$.
Specifically, we demonstrate that for \emph{any} $n$-point metric $M$,
the weight of the MST of every connected SDG for $M$
is $O(\log n) \cdot w(MST(M))$. 
\end{abstract}
\end{titlepage}

\pagenumbering {arabic} 

\section{Introduction}
\vspace{-0.04in}
{\bf 1.1~ The MST of Symmetric Disk Graphs.~}
Consider a network that is represented as an (undirected) weighted graph $G = (V,E,w)$,
and assume that we want to compute a spanning tree for $G$ of \emph{small weight},
i.e., of weight that is close to the weight $w(MST(G))$ of the minimum spanning tree (MST) for $G$.
(See Section 1.6 for the definition of weight.)
However, due to some physical constraints (e.g., network faults) we are 
only given a connected spanning subgraph $G'$ of $G$, rather than $G$ itself. 
In this situation it is natural to use the MST of
the given subgraph $G'$.
The \emph{weight-coefficient} of $G'$ with respect to $G$
is defined as the ratio between $w(MST(G'))$ and $w(MST(G))$.
If the weight-coefficient of $G'$ is small enough, we can use $MST(G')$ as 
a spanning tree for $G$ of small weight. 

The problem of computing spanning trees of small weight (especially the MST) is a fundamental one in
Computer Science \cite{KRY94,KKT95,Chazelle00,PR02,Elkin06,CS09},
and the above scenario arises naturally in many practical contexts (see, e.g., \cite{Salowe91,DF94,ZSN02,Li03,LWWF04,LWS04,DPP06,DPP062}).
In particular, this scenario is motivated 
by wireless network design.

In this paper we focus on the \emph{symmetric disk graph model} in wireless communication networks,
which 
has been subject to considerable research.
(See \cite{HK01,EJS01,FFF04,LL06,TD06,CFKP07,TWLZD07,AACK10}, and the references therein.) 
Let $M = (V,\delta)$ be an $n$-point metric that is represented as 
a complete weighted graph $G(M) = (V, {V \choose 2}, w)$
in which the weight $w(e)$ of each edge $e = (u,v)$  is equal to $\delta(u,v)$.
Also, let $r: V \rightarrow \mathbb R^+$ be a transmission range assignment
that maps each point $v \in V$ to the disk of radius $r(v)$ around it.
The \emph{symmetric disk graph} (henceforth, SDG) that corresponds to $M$ and $r$,
denoted $SDG(M,r)$,
\footnote{The definition of symmetric disk graph can be generalized in the obvious way for any weighted graph.
Specifically, the \emph{symmetric disk graph} $SDG(G,r)$ that corresponds to a weighted graph $G = (V,E,w)$ and a range assignment $r$ is
the undirected spanning subgraph of $G$ whose edge set includes an edge $e=(u,v) \in E$  if both $r(u)$ and $r(v)$ are no smaller than $w(e)$.}
is the undirected spanning subgraph of $G(M)$ whose edge set includes
an edge $e=(u,v)$ if both $r(u)$ and $r(v)$ are no smaller than $w(e)$. 
Under the symmetric disk graph model we cannot use all the edges of $G(M)$,    
but rather only those that are present in $SDG(M,r)$.
Clearly, if $r(v) \ge diam(M)$\footnote{The \emph{diameter} of a metric $M$, denoted $diam(M)$, is defined as the largest pairwise distance in $M$.} for each point $v \in V$, then $SDG(M,r)$ is simply the complete graph $G(M)$. 
However, the transmission ranges are usually significantly shorter than $diam(M)$, and many edges that belong to $G(M)$
may not be present in $SDG(M,r)$. 
Therefore, it is generally impossible to use the MST of $M$ under the symmetric disk graph model,
simply because some of the edges of $MST(M)$ are not present in $SDG(M,r)$ and thus cannot be accessed.
Instead, assuming the weight-coefficient of 
$SDG(M,r)$ with respect to $M$ is small enough, we can use $MST(SDG(M,r))$ as a spanning tree for $M$ of small weight.

Abu-Affash et al.\ \cite{AACK10}
showed that for any \emph{2-dimensional Euclidean} $n$-point metric $M$, the weight of the MST of every \emph{connected} SDG for $M$
is $O(\log n) \cdot w(MST(M))$. In other words, they proved that
for any 2-dimensional Euclidean $n$-point metric,
the weight-coefficient of every \emph{connected} SDG 
 is $O(\log n)$.
In addition, Abu-Affash et al.\ \cite{AACK10}  provided a matching lower bound of $\Omega(\log n)$ on the weight-coefficient
of connected SDGs
that applies to a basic 1-dimensional Euclidean metric.
Notably, the upper bound proof of \cite{AACK10} relies heavily on basic geometric properties of 2-dimensional Euclidean metrics,
and does not extend to higher dimensions.
A natural question that arises is whether the logarithmic upper bound of 
\cite{AACK10} on the weight-coefficient of connected SDGs can be generalized for wider families of metrics, such as 3-dimensional Euclidean metrics.

In this paper we generalize the upper bound of Abu-Affash et al.\ \cite{AACK10} 
for Euclidean metrics of any dimension. Furthermore, our upper bound extends
to \emph{arbitrary metrics} and, in particular, it applies to any of the normed spaces $\ell_p$.
Specifically, we demonstrate that for \emph{any} $n$-point metric $M$,
every connected SDG has weight-coefficient $O(\log n)$.
In fact, our upper bound is even more general,
applying to unconnected SDGs as well. That is, we show that the weight
of the minimum spanning forest (MSF) of every (possibly unconnected) SDG for $M$
is $O(\log n) \cdot w(MST(M))$.

The fact that the 
weight-coefficient of SDGs for arbitrary metrics
is relatively small
is quite surprising. 
In particular, we demonstrate that for other basic parameters of spanning trees,
the situation 
is fundamentally different.
Consider, for example, the maximum degree (henceforth, degree) parameter. Clearly, for any metric there
is a spanning tree with degree 2. On the other hand, consider an $n$-point metric $M^*$ in which the distance between
a designated point $rt \in M^*$ and every other point is equal to 1, and all other distances
are equal to 2. The SDG corresponding to
$M^*$ and the range assignment $r \equiv 1$ that maps each point to the unit disk around it is the $n$-star graph rooted
at $rt$, having degree $n-1$. Thus, the \emph{degree-coefficient}
of connected SDGs for metrics
can be as large as $\Omega(n)$ in general, which is \emph{exponentially larger} than the weight-coefficient. (See Section \ref{appb} for the formal definition
of degree-coefficient.)
We show that the same lower bound of $\Omega(n)$ also applies to other parameters of spanning trees, including radius and depth,
 diameter and hop-diameter,
sum of all pairwise distances, and sum of all distances from a designated vertex. 

Finally, we remark that our logarithmic upper bound on the weight-coefficient of SDGs
does not extend to general (undirected) weighted graphs.
Indeed, if the weight function of the graph does not satisfy the triangle inequality,
the weight-coefficient of SDGs can be arbitrarily large even for complete  graphs.
Moreover, we demonstrate that there are (non-complete) 1-dimensional Euclidean $n$-vertex graphs\footnote{A \emph{1-dimensional Euclidean 
graph} is a weighted graph in which the vertices represent points on a line, and the weight of each edge is
equal to the Euclidean distance between its endpoints.} for which the weight-coefficient of SDGs can be
as large as $\Omega(n)$. 
\vspace{0.06in}\\
{\bf 1.2~ The Range Assignment Problem.~}
Given a network 
$G = (V,E,w)$,
a \emph{range assignment} for $G$ is an assignment of transmission ranges to each of the vertices of $G$.
A range assignment is called \emph{complete} if the induced (directed) communication
graph is strongly connected. In the \emph{range assignment problem} the objective is to find a complete
range assignment for which the total power consumption (henceforth, cost) is minimized. The power consumed by a vertex
$v \in V$ is $r(v)^\alpha$, where $r(v)$ is the range assigned to $v$ and $\alpha \ge 1$ is some constant. 
Thus the cost of the range assignment is given by $\sum_{v \in V}r(v)^{\alpha}$.
The
range assignment problem was first studied by Kirousis et al.\ \cite{KKKP00},
who proved that the problem is NP-hard in 3-dimensional Euclidean metrics,
assuming $\alpha = 2$, and also presented a simple 2-approximation algorithm.     
Subsequently, Clementi et al.\ \cite{CPS99} proved that the problem
remains NP-hard in 2-dimensional Euclidean metrics.

We believe that it is more realistic to study the range assignment problem under the symmetric disk graph model.
Specifically, the potential transmission range of a vertex $v$ is bounded by some maximum range $r(v)$, and any two 
vertices $u,v$ can directly communicate with each other if and only if $v$ lies within the range assigned to $u$
and vice versa. Blough et al.\ \cite{BLRS02} showed that this version of the range assignment problem
is also NP-hard in 2-dimensional and 3-dimensional Euclidean metrics.
Also, Calinescu et al.\ \cite{CMZ02} devised a $(1+\frac{1}{2}\ln 3 + \epsilon)$-approximation scheme and a more practical
$(\frac{15}{8})$-approximation algorithm. Abu-Affash et al.\ \cite{AACK10} showed that, assuming $\alpha = 1$,
the cost of an optimal range assignment with bounds on the ranges is greater by at most a logarithmic factor
than the cost of an optimal range assignment without such bounds. 
This result of Abu-Affash et al.\ \cite{AACK10} is a simple corollary of their upper bound 
on the weight-coefficient of SDGs for 2-dimensional Euclidean metrics.
Consequently, this result of \cite{AACK10} for the range assignment problem  holds only in 2-dimensional Euclidean metrics.
By applying our generalized upper bound on the weight-coefficient of SDGs, we extend this result of Abu-Affash et al.\ \cite{AACK10}
to arbitrary metrics.
\vspace{0.06in}\\
{\bf 1.3~ Proof Overview.~}
As was mentioned above, the upper bound proof of \cite{AACK10} is very specific, and relies heavily on 
basic geometric properties of 2-dimensional Euclidean metrics. Hence, it does not apply to 3-dimensional Euclidean metrics,
let alone to arbitrary metrics. Our upper bound proof is based on completely different principles.
In particular, it is independent of the geometry of the metric and applies to every complete graph whose weight function
satisfies the triangle inequality.
In fact, at the heart of our proof 
is a lemma that applies to an even wider family of graphs, namely, the family of all traceable\footnote{A graph is called
\emph{traceable} if it contains a Hamiltonian path.} weighted graphs.
Specifically, let $S$ and $H$ be an SDG and a minimum-weight Hamiltonian path of some traceable weighted $n$-vertex graph $G$, respectively,
and let $F$ be the MSF of $S$. Our lemma states that there is a set $\tilde E$ of edges in $F$ of weight at most $w(H)$,
such that the graph $F \setminus \tilde E$ obtained by removing all edges of $\tilde E$ from $F$ contains at least $\frac{1}{5} \cdot n$ isolated vertices. 
The proof of this lemma is based
on a delicate combinatorial argument that does not assume either that the graph $G$ is complete or that its weight function 
satisfies the triangle inequality.
We believe that this lemma is of independent interest. (See Lemma \ref{key} in Section \ref{sec3}.)
By employing this lemma inductively, we are able to show that the weight of $F$ is bounded above by $\log_{\frac{5}{4}} n \cdot w(H)$,
which, by the triangle inequality, yields an upper bound of $2 \cdot \log_{\frac{5}{4}} n$ on the weight-coefficient of $S$
 with respect to $G$.
Interestingly, our upper bound of $2 \cdot \log_{\frac{5}{4}} n$ on the weight-coefficient of SDGs for arbitrary metrics
improves  the corresponding upper bound of 
\cite{AACK10}  (namely, $90 \cdot \log_{\frac{5}{4}} n + 1$), which holds only in 2-dimensional Euclidean metrics, 
by a multiplicative factor of 45.
\vspace{0.06in}
\\
{\bf 1.4~ Related Work on Disk Graphs.~}
The symmetric disk graph model is a generalization
of the extremely well-studied \emph{unit disk graph model} (see, e.g., \cite{CCJ90,Li03,KMPS04,LWS04,vanL05}). 
The \emph{unit disk graph} of a metric $M$, denoted $UDG(M)$, is the 
symmetric disk graph corresponding
to $M$ and the range assignment $r \equiv 1$ that maps each point to the unit disk around it.
(It is usually assumed that $M$ is a 2-dimensional Euclidean metric.)
It is easy to see that in the case when $UDG(M)$ is connected, all edges of $MST(M)$ belong to $UDG(M)$,
and so $MST(UDG(M)) = MST(M)$. 
Hence the weight-coefficient of connected unit disk graphs for arbitrary metrics 
is equal to 1. In the general case, we note that all edges of $MSF(UDG(M))$ belong to $MST(M)$, and so the weight-coefficient
of (possibly unconnected) unit disk graphs for arbitrary metrics is at most 1.

Another model that has received much attention in the literature
is the \emph{asymmetric disk graph model} (see, e.g., \cite{KMPS04,TTD08,PR10,PR102,AACK10}).
The \emph{asymmetric disk graph} 
corresponding to a metric $M = (V,\delta)$ and a range assignment $r$
is the directed graph over $V$, where there is an arc of weight $\delta(u,v)$ from $u$ to $v$ if $r(u) \ge \delta(u,v)$.
On the negative side, Abu-Affash et al.\ \cite{AACK10} provided a lower bound of $\Omega(n)$ on the weight-coefficient
of strongly connected asymmetric disk graphs that applies to a 1-dimensional Euclidean $n$-point metric. 
However, asymmetric communication models are generally considered to be
impractical, because in such models many
communication primitives become unacceptably complicated  \cite{Parkash99,Watten05}.
In particular, the asymmetric disk graph model is often viewed as 
less realistic than the symmetric disk graph model,
where, as was mentioned above,
we obtain
a logarithmic upper bound on the weight-coefficient for arbitrary metrics.
\vspace{0.06in}\\
{\bf 1.5~ Structure of the Paper.~}
The main result of this paper is given in Section \ref{sec3}. 
Therein we obtain a logarithmic upper bound on the weight-coefficient of SDGs for arbitrary metrics.
In Section \ref{sec5} we provide an application of this upper bound 
to the range assignment problem.
%
Finally, Section \ref{sec4} is devoted to negative results concerning two possible extensions of the upper bound 
of Section \ref{sec3}.
We start (Section \ref{appa}) with showing that this 
upper bound 
does not extend to general graphs, and proceed (Section \ref{appb})
showing that in comparison to the weight parameter, 
other basic parameters of spanning trees  incur significantly larger bounds.    
\vspace{0.06in}\\
{\bf 1.6~ Preliminaries.~}
Given a (possibly weighted) graph $G$, its vertex set (respectively, edge set) is denoted by $V(G)$ (resp., $E(G)$).
For an edge set $E' \subseteq E(G)$, we denote by $G \setminus E'$ the graph
obtained by removing all edges of $E'$ from $G$. Similarly, for an edge set $E''$ over the vertex set $V(G)$, we
denote by $G \cup E''$ the graph obtained by adding all edges of $E''$ to $G$.
The weight of an edge $e$ in $G$ is denoted by $w(e)$. 
For an edge set $E \subseteq E(G)$, its weight $w(E)$ is defined as the sum of all edge weights in $E$, i.e., $w(E)
= \sum_{e \in E} w(e)$. The weight of $G$ is defined as the weight of its edge set $E(G)$, namely,
$w(G) = w(E(G))$.
Finally, for a positive integer $n$, we denote the set $\{1, 2,\ldots , n\}$ by  $[n]$.

\section{The MST of SDGs is Light} \label{sec3}
In this section we prove that the weight-coefficient of SDGs for arbitrary $n$-point metrics is $O(\log n)$.

We will use the following well-known fact in the sequel.  
\begin{fact} \label{maxweight}
Let $G$ be a weighted graph in which ell edge weights are distinct. Then $G$ has a unique MSF,
and the edge of maximum weight in every cycle of $G$ does not belong to the MSF of $G$.   
\end{fact}
\ignore{
\begin{proof}
Denote by $E^*(C)$ the set of edges of maximum weight in $C$, and suppose for contradiction
that all edges of $E^*(C)$ belong to $F$. Consider the tree $T_C$ of $F$ that contains the vertex
set of the cycle $C$, and note that all edges of $E^*(C)$ belong to $T_C$, i.e., $E^*(C) \subseteq E(T_C)$. 
Removing an arbitrary edge $e^*$ of $E^*(C)$ decomposes
the tree $T_C$ into two subtrees. Since the remainder $C \setminus \{e^*\}$ of $C$ reconnects these two subtrees, 
there must be an edge $e$ of $C$ that does not belong to $T_C$, with endpoints in different subtrees.
Adding this edge reconnects the two subtrees, thus forming a new spanning forest $F' = F \setminus \{e^*\} \cup \{e\}$ of $G$.
However, since $e$ does not belong to $T_C$ and $E^*(C) \subseteq E(T_C)$, it follows that $w(e) < w(e^*)$.
We conclude that the weight $w(F')$ of the new spanning forest $F'$ of $G$ satisfies $w(F') = w(F) - w(e^*) + w(e) < w(F)$,
a contradiction. \qed
\end{proof}
}

In what follows we assume for simplicity that all the distances in any metric 
are distinct.
This assumption does not lose generality, since any ties
can be broken using, e.g., lexicographic rules. We may henceforth assume 
that there is a unique MST for any metric, and a unique MSF for every SDG of any metric.


The following lemma is central in our upper bound proof.
\begin{lemma} \label{key}
Let $M=(V,\delta)$ be an $n$-point metric and let $r: V \rightarrow \mathbb R^+$ be a range assignment.
Also, let $F = (V,E_F)$ be the MSF of the symmetric disk graph $S = SDG(M,r)$
and let $H = (V,E_{H})$ be a minimum-weight Hamiltonian path of $M$.
Then there is an edge set $\tilde E \subseteq E_F$
of weight at most $w(H)$, such that the graph $F \setminus \tilde E$ 
contains at least $\frac{1}{5} \cdot n$
isolated vertices.
\end{lemma}
{\bf Remark:} This statement remains valid if instead of the metric $M$ we take a general traceable weighted graph.
\begin{proof}
Denote by $E'$ the set of edges in $H$ that belong to the SDG $S$, i.e., $E' = E_{H} \cap E(S)$,
and let $E'' = E_{H} \setminus E'$ be the complementary edge set of $E'$ in $E_H$.
Also, denote by $E'_1$ the set of edges in $E'$ that belong to the MSF $F$, i.e., $E'_1 = E' \cap E_F$, and let
$E'_2 = E' \setminus E'_1$ be the complementary edge set of $E'_1$ in $E'$. Note that
(1) $E' \subseteq E(S)$, (2) $E'' \cap E(S) = \emptyset$, (3) $E'_1 \subseteq E_F$,
and (4) $E'_2 \cap E_F = \emptyset$.


Write $F_0 = F, k = |E'_2|$, and let $e'_1,e'_2,\ldots,e'_k$ denote the edges of $E'_2$
by increasing order of weight. Next, we construct $k$ spanning forests $F_1,F_2,\ldots,F_k$
of $S$ in the following iterative process.
For each index $i = 1,2,\ldots,k$, the graph $F_{i-1} \cup \{e'_{i}\}$ obtained from $F_{i-1}$
by adding to it the edge $e'_{i}$ contains a unique cycle $C_i$. 
Since $H$ is cycle-free, at least one edge of $C_i$
does not belong to $H$.  Let $\tilde e_{i}$ be an arbitrary such edge,
and denote by $F_{i} = F_{i-1} \cup \{e'_{i}\} \setminus \{\tilde e_{i}\}$ the graph  obtained from $F_{i-1}$
by adding to it the edge $e'_{i}$ and removing the edge $\tilde e_{i}$. 
\\It is easy to verify that the cycles $C_1,C_2,\ldots,C_k$ that are identified during
this process are subgraphs of the symmetric disk graph $S$. Moreover, for each index $i \in [k]$,
we have $E(C_i) \subseteq E_F \cup E'_2$,
 yielding $\tilde e_i \in E_F \setminus E_H$.
\\Write $\tilde E_2 = \{\tilde e_1,\tilde e_2,\ldots, \tilde e_k\}$,
and notice
that $|\tilde E_2| = |E'_2|, \tilde E_2 \subseteq E_F\setminus E_H$.
\\The following claim implies that $w(\tilde E_2) \le w(E'_2)$.
\begin{claim} \label{firstc}
For each index $i \in [k]$, $w(\tilde e_i) \le w(e'_i)$.
\end{claim}
\begin{proof}
Fix an arbitrary index $i \in [k]$, and define $E'_{(i)} = \{e'_1,\ldots,e'_{i}\}$. 
Recall that the cycle $C_i$ is a subgraph of $S$,
and notice that each edge of $C_i$ that do not belong to $F$ must belong
to $E'_{(i)}$, i.e., $E(C_i) \setminus E_F \subseteq E'_{(i)}$.  
Fact \ref{maxweight} implies that the edge of maximum weight in  $C_i$, denoted $e^*_i$, 
does not belong to $F$,
and so $e^*_i \in E'_{(i)}$.
Since $e'_i$ is the edge of maximum weight in $E'_{(i)}$, 
it holds that $w(e^*_i) \le w(e'_i)$.
Also, as $\tilde e_i$ belongs to $C_i$, we have by definition $w(\tilde e_i) \le w(e^*_i)$. 
Claim \ref{firstc} follows.
\qed
\end{proof}

Denote by $H'' = H \setminus E' = H \setminus (E'_1 \cup E'_2)$ and $F'' = F \setminus (E'_1 \cup \tilde E_2)$ 
 the graphs obtained from $H$ and $F$
by removing all edges of $E' = E'_1 \cup E'_2$ and $E'_1 \cup \tilde E_2$, respectively. By definition, $E(H'') = E''$.
For an edge $e = (u,v)$, denote by $min(e)$ the endpoint of $e$ with smaller radius, i.e.,
$min(e) = u$ if $r(u) < r(v)$, and $min(e) = v$ otherwise. Consider an arbitrary edge $e \in E''$.
Since no edge of $E''$ belongs to the symmetric disk graph $S$,
it follows that $r(min(e)) < w(e)$.      
Also, since the graph $F''$ is a subgraph of $S$, 
the weight of every edge that is incident to $min(e)$ in $F''$ is no greater than $r(min(e))<w(e)$.

Next, we remove some edges of the graphs $H''$ and $F''$ and add them
to the two initially empty edge sets $E^*_H$ and $E^*_F$, respectively.
This is done in the following way.
We initialize $E^*_H = E^*_F = \emptyset$,
and then examine the edges of $E''$ one after another in an arbitrary order. 
For each edge $e \in E''$, we check
whether the vertex $min(e)$ is isolated in $F''$ or not. If $min(e)$ is isolated in $F''$, we leave
$H'',F'',E^*_H$, and $E^*_F$ intact.
Otherwise, at least one edge is incident to $min(e)$ in $F''$.
Let $\tilde e$ be an arbitrary such edge,
and note that $w(\tilde e) \le r(min(e)) < w(e)$.
We remove the edge $e$ from the graph $H''$ and add it to the edge set $E^*_H$, and remove the edge $\tilde e$ from the graph
$F''$ and add it to the edge set $E^*_F$. 
This process is repeated iteratively until all edges of $E''$ have been examined. 
Notice that at each stage of this process, it holds that $|E^*_F| = |E^*_H|, E^*_F \subseteq E_F \setminus E_H, w(E^*_F) \le w(E^*_H)$.

In what follows we consider the graphs $H'',F''$ and edge sets $E^*_H,E^*_F$
resulting at the end of this process.
Define $\tilde E = E'_1 \cup \tilde E_2 \cup E^*_F$.
By construction, we have $\tilde E \subseteq E_F,
F'' = F \setminus \tilde E$.
The following claim completes the proof
of Lemma \ref{key}.
\begin{claim} \label{secondc}
~(1) $w(\tilde E) \le w(H)$.
~~(2) The graph $F'' = F \setminus \tilde E$ contains at least $\frac{1}{5} \cdot n$ isolated vertices.
\end{claim}
\begin{proof}
We start with proving the first assertion of the claim.
Define $E = E' \cup E^*_H = E'_1 \cup E'_2 \cup E^*_H$.   
It is easy to see that 
the edge sets $E'_1, E'_2$, and $E^*_H$ are pairwise disjoint subsets of $E_H$, yielding $w(E'_1) + w(E'_2) + w(E^*_H)
= w(E) \le w(H)$.
Recall that $w(\tilde E_2) \le w(E'_2)$
and $w(E^*_F) \le w(E^*_H)$.
Consequently,
\begin{eqnarray*}
w(\tilde E)  ~&=&~ w\left(E'_1 \cup \tilde E_2 \cup E^*_F\right) ~\le~ w(E'_1) + w(\tilde E_2) + w(E^*_F) 
\\ ~&\le&~ w(E'_1) + w(E'_2) + w(E^*_H) ~=~ w(E) ~\le~ w(H).
\end{eqnarray*}
Next, we prove the second assertion of the claim.
Denote by $m_{H}$ (respectively, $m_{F}$) the number $|E(H'')|$ (resp., $|E(F'')|$) of edges in the graph $H''$ (resp., $F''$).
\\Suppose first that $m_F < \frac{2}{5} \cdot n$.
Observe that in any $n$-vertex graph with $m$ edges
there are at least $n- 2m$ isolated vertices. Thus, the number of isolated vertices in $F''$ is
bounded below by $n-2m_F > n - \frac{4}{5} \cdot n = \frac{1}{5} \cdot n$, as required.
\\We henceforth assume that $m_F \ge \frac{2}{5} \cdot n$.
\\Recall that $|\tilde E_2| = |E'_2|, |E^*_F| = |E^*_H|$. It is easy to see that the edge sets
$E'_1,\tilde E_2,E^*_F,E'_2$, and $E^*_H$ are pairwise disjoint, 
and so
$$|\tilde E| ~=~ |E'_1| + |\tilde E_2| + |E^*_F| ~=~ |E'_1| + |E'_2| + |E^*_H| ~=~ |E|.$$
Observe that $H'' = H \setminus E$ and recall that  $F'' = F \setminus \tilde E$,
yielding $|E(H'')|  ~=~ |E_H| - |E|,  |E(F'')| = |E_F| - |\tilde E|$.
Also, note that $|E_H| = n-1 \ge |E_F|$.
Altogether, \begin{eqnarray} \label{omit}
  m_H ~=~ |E(H'')|  ~=~ |E_H| - |E|  ~\ge~
|E_F| - |\tilde E|  ~=~ |E(F'')| ~=~ m_F.
\end{eqnarray}
Next, observe that for each edge $e$ in $H''$, the vertex $min(e)$ is
isolated in $F''$. By definition, for any pair $e,e'$ of non-incident edges in $H''$, $min(e) \ne min(e')$.
Since the graph $H''$ is cycle-free and the maximum degree of a vertex in $H''$ is at most two,
it follows that the number $m_H$ of edges in $H''$ is at most twice greater than the
number of isolated vertices in $F''$.
Thus, 
the number of isolated vertices in
 $F''$ is bounded below by $\frac{1}{2} \cdot {m_H} \ge \frac{1}{2} \cdot {m_F} \ge \frac{1}{5} \cdot n$. (The first inequality
 follows from (\ref{omit}) whereas the second inequality follows from the above assumption.)
\\ This completes the proof of claim \ref{secondc}.
\qed
\end{proof}
Lemma \ref{key} follows. \qed
\end{proof}

Next, we employ Lemma \ref{key} inductively to upper bound the weight of SDGs in terms of the weight
of the minimum-weight Hamiltonian path of the metric. The desired upper bound of $O(\log n)$ on the
weight-coefficient
of SDGs for arbitrary $n$-point metrics would immediately follow.
\begin{lemma}
Let $M=(V,\delta)$ be an $n$-point metric and let $r: V \rightarrow \mathbb R^+$ be a range assignment.
Also, let $F = (V,E_F)$ be the MSF of 
the symmetric disk graph 
$S = SDG(M,r)$
and let $H = (V,E_{H})$ be a minimum-weight Hamiltonian path of $M$.
Then $w(F) \le \log_{\frac{5}{4}} n \cdot w(H)$.
\end{lemma}
\begin{proof}
The proof is by induction on the number $n$ of points in the metric $M$.
\\\emph{Basis: $n \le 4$.}
The case $n=1$ is trivial. 
Suppose next that $2 \le n \le 4$. In this case $\log_{\frac{5}{4}} n \ge \log_{\frac{5}{4}} 2 > 3$. Also, the MSF $F$ of $S$ contains at most 3 edges.
By the triangle inequality, the weight of each edge of $F$ is bounded above by the weight $w(H)$ of the Hamiltonian path $H$.
Hence, $w(F) \le 3 \cdot w(H) < \log_{\frac{5}{4}} n \cdot w(H)$.
\\\emph{Induction step:} We assume that the statement holds for all smaller values of $n$, $n \ge 5$,
and prove it for $n$. By Lemma \ref{key}, there is an edge set $\tilde E \subseteq E_F$ of weight
at most $w(H)$, such that the set $I$ of isolated vertices in the graph $F \setminus \tilde E$ satisfies
$|I| \ge \frac{1}{5} \cdot n$. Consider the complementary edge set $\hat E = E_F \setminus \tilde E$ of edges
in $F$. Observe that no edge of $\hat E$ is incident to a vertex of $I$. Denote by $\hat M$ the
sub-metric of $M$ induced by the point set of $\hat V = V \setminus I$, 
let $\hat S = SDG(\hat M,r)$ 
be the SDG corresponding to $\hat M$ and the original range assignment $r$,
and let $\hat F = (\hat V,E_{\hat F})$ be the MSF of $\hat S$. 
Notice that the induced subgraph of $S$ over the vertex set $\hat V$ is equal to $\hat S$,
and so $\hat E$ is a subgraph of $\hat S$.
Since $\hat F$ is a spanning forest of $\hat S$, replacing the edge set $\hat E$
of $F$ by the edge set $E_{\hat F}$ does not affect the connectivity of the graph,
i.e., 
the graph $\bar F = F \setminus \hat E \cup E_{\hat F}$ 
that is obtained from $F$ by removing the 
edge set $\hat E$ and adding the edge set $E_{\hat F}$ has
exactly the same connected components as $F$.
Thus, by breaking all cycles in the graph $\bar F$, we get a spanning forest 
of $S$. The weight of this spanning forest is bounded above by the weight $w(\bar F) = w\left(F \setminus \hat E \cup E_{\hat F}\right)$ 
of the graph $\bar F$, and is bounded below
by the weight $w(F)$ of the MSF $F$ of $S$.
It follows that
$w(\hat E) \le w(E_{\hat F}) = w(\hat F)$. Write $\hat n = |\hat V|$, and let $\hat H = (\hat V,E_{\hat H})$ be a
minimum-weight Hamiltonian path of $\hat M$. 
Since $|I| \ge \frac{1}{5} \cdot n$, we have $$\hat n ~=~ |\hat V| ~=~ |V \setminus I| ~\le~ \frac{4}{5} \cdot n ~\le~ n-1.$$
(The last inequality holds for $n \ge 5$.)
By the induction hypothesis for $\hat n$, $w(\hat F) \le \log_{\frac{5}{4}} \hat n \cdot w(\hat H)$.
Also, the triangle inequality implies that $w(\hat H) \le w(H)$.
Hence, 
\begin{eqnarray*}
\nonumber w(\hat E) ~&\le&~ w(\hat F) ~\le~ \log_{\frac{5}{4}} \hat n \cdot w(\hat H) ~\le~ \log_{\frac{5}{4}} \left(\frac{4}{5} \cdot n\right) \cdot w(H)
\\ ~&=&~ \log_{\frac{5}{4}} n \cdot w(H) - w(H).
\end{eqnarray*}
We conclude that
\begin{eqnarray*}
w(F) ~&=&~ w(E_F) ~=~ w(\tilde E) + w(E_F \setminus \tilde E) ~=~ w(\tilde E) + w(\hat E)
\\ ~&\le&~ w(H) + \log_{\frac{5}{4}} n \cdot w(H) - w(H) ~=~ \log_{\frac{5}{4}} n \cdot w(H). 
\end{eqnarray*}
\qed
\end{proof}

By the triangle inequality, the weight of the minimum-weight Hamiltonian path of any metric 
is at most twice greater than the weight of the MST for that metric. We derive the main result of this paper.
\begin{theorem}\label{final}
For any $n$-point metric $M=(V,\delta)$ and any range assignment $r: V \rightarrow \mathbb R^+$,
$w(MSF(SDG(M,r))) = O(\log n) \cdot w(MST(M))$.
\end{theorem}

\section{The Range Assignment Problem} \label{sec5}
In this section we demonstrate that for any metric, the cost of an optimal range assignment with bounds on the ranges 
is greater by at most a logarithmic factor than the cost of an optimal range assignment without such bounds. This result follows as a simple corollary
of the upper bound given in Theorem \ref{final}.


Let $M = (V,\delta)$ be an $n$-point metric, 
and assume that the $n$ points of $V$, denoted by $v_1,v_2,\ldots,v_n$, represent transceivers.
Also, let $r': V \rightarrow \mathbb R^+$ be a function that provides a maximum transmission range for each of the points of $V$.
In the \emph{bounded range assignment problem} 
the objective is to compute a range assignment $r: V \rightarrow \mathbb R^+$, such that
(i) for each point $v_i \in V$, $r(v_i) \le r'(v_i)$, (ii) the induced SDG (using the ranges
$r(v_1),r(v_2),\ldots,r(v_n)$), namely $SDG(M,r)$, is connected, and (iii) $\sum_{i=1}^n r(v_i)$ is minimized. 
In the \emph{unbounded range assignment problem} the maximum transmission range for each of the points of $V$
is unbounded, i.e., $r'(v_i) = Diam(M)$, for each point $v_i \in V$. 
The function $r'$ is called a \emph{bounding function}.   
Also, the sum $\sum_{i=1}^n r(v_i)$ is called the \emph{cost} of the range assignment $r$, and is denoted by $COST(r)$. 

Fix an arbitrary bounding function $r': V \rightarrow \mathbb R^+$.
Denote by $OPT(M,r')$ the cost of an optimal 
solution for the
bounded range assignment problem corresponding to $M$ and $r'$.
Also, denote by $OPT(M)$ the
cost of an optimal 
solution for the unbounded range assignment problem corresponding to $M$.
Clearly, $OPT(M,r') \ge OPT(M)$.
Next, we show that $OPT(M,r') = O(\log n) \cdot OPT(M)$.


Let $SDG(M,r')$ be the SDG corresponding to $M$ and $r'$,
and let $T$ be the MST of $SDG(M,r')$. We define $r$ to be the range assignment that assigns $r(v_i)$ with the weight of the heaviest
edge incident to $v_i$ in $T$, for each point $v_i \in V$. By construction, $r(v_i) \le r'(v_i)$, for each point $v_i \in V$.
Also, notice that the SDG corresponding to $M$ and $r$, namely $SDG(M,r)$, contains $T$ and in thus connected. 
Hence, the range assignment $r$ provides a feasible solution for the bounded range assignment problem corresponding to $M$ and $r'$, yielding
$OPT(M,r') \le COST(r)$. By a double counting argument, it follows that $\sum_{i=1}^n r(v_i) ~\le~ 2 \cdot w(T)$.
Also, by Theorem \ref{final}, $w(T) = w(MST(SDG(M,r'))) = O(\log n) \cdot w(MST(M))$.
Finally, it is easy to verify that $w(MST(M)) \le OPT(M)$.
Altogether,
\begin{eqnarray*}
\nonumber OPT(M,r') ~&\le&~ COST(r) ~=~ \sum_{i=1}^n r(v_i) ~\le~ 2 \cdot w(T) 
\\ ~&=&~ O(\log n) \cdot w(MST(M))
~=~ O(\log n) \cdot OPT(M).
\end{eqnarray*}






\section{Negative Results} \label{sec4}
\subsection{The MST of SDGs in General Graphs is Heavy} \label{appa}
In this section we show that the upper bound of Theorem \ref{final}
on the weight-coefficient of SDGs
does not extend to general (undirected) weighted graphs. 

Let $W$ be an arbitrary large number. 

Suppose first that the weight function of the graph does not satisfy the triangle inequality.
Let $C_3$ be the 3-cycle graph on the vertex set $\{a,b,c\}$ in which the edge $(a,b)$ has unit weight, the edge
$(a,c)$ has weight 2, and the edge $(b,c)$ has weight $W$. 
Observe that the SDG corresponding to $C_3$ and the range assignment $r$ that assigns $r(a) = 1, r(b) = r(c) = W$, denoted $SDG(C_3,r)$,
contains the two edges $(a,b)$ and $(b,c)$ but does not contain the edge $(a,c)$.
Thus, $SDG(C_3,r)$ is a spanning tree of $C_3$ of weight $W+1$. On the other hand,
the MST of the original graph $C_3$ consists of the two edges $(a,b)$ and $(a,c)$ and has weight 3. 
It follows that the weight-coefficient of $SDG(C_3,r)$ with respect to $C_3$ is $\Omega(W)$.
In other words, the weight-coefficient of SDGs can be arbitrary large in general.

When the weight function of the graph satisfies the triangle inequality, the weight of each edge
in the MST (or the MSF) of the SDG is bounded above by the weight of the entire MST of the original graph.
Hence, the weight-coefficient of SDGs in this case is at most linear in the number of vertices in the graph. 
\\Next, we provide a matching lower bound that applies to (non-complete) 1-dimensional Euclidean graphs.
Let $\epsilon > 0$ be a tiny number, and consider the set $V$ of $n$ points 
$s,u_1,u_2,u_3,\ldots,u_{n-2},t$ that lie on the $x$-axis
with coordinates $0,1, 1+\epsilon,1+2\epsilon,\ldots,1+(n-3) \cdot \epsilon,W+1$, respectively, where $W \gg n$ and $\epsilon \ll \frac{1}{n}$. 
Define $U = V \setminus \{s,t\} =  \{u_1,u_2,\ldots,u_{n-2}\}$,     
and observe that the Euclidean distance
between $s$ (respectively, $t$) and each point of $U \setminus \{u_1\}$ is slightly larger than 1 (resp., slightly smaller than $W$),
whereas the Euclidean distance between $s$ (respectively, $t$) and $u_1$ is equal to 1 (resp., $W$).
Let $G = (V,E,\|\cdot\|)$ be the Euclidean graph over the point set $V$, where 
$E = \{(s,u_i) ~\vert~ u_i \in U\}
 \cup \{(u_i,t) ~\vert~ u_i \in U\}$. (See Figure \ref{1dim} for an illustration.)
\begin{figure*}[htp]
\begin{center}
\begin{minipage}{\textwidth} 
\begin{center}
\setlength{\epsfxsize}{6.7in} \epsfbox{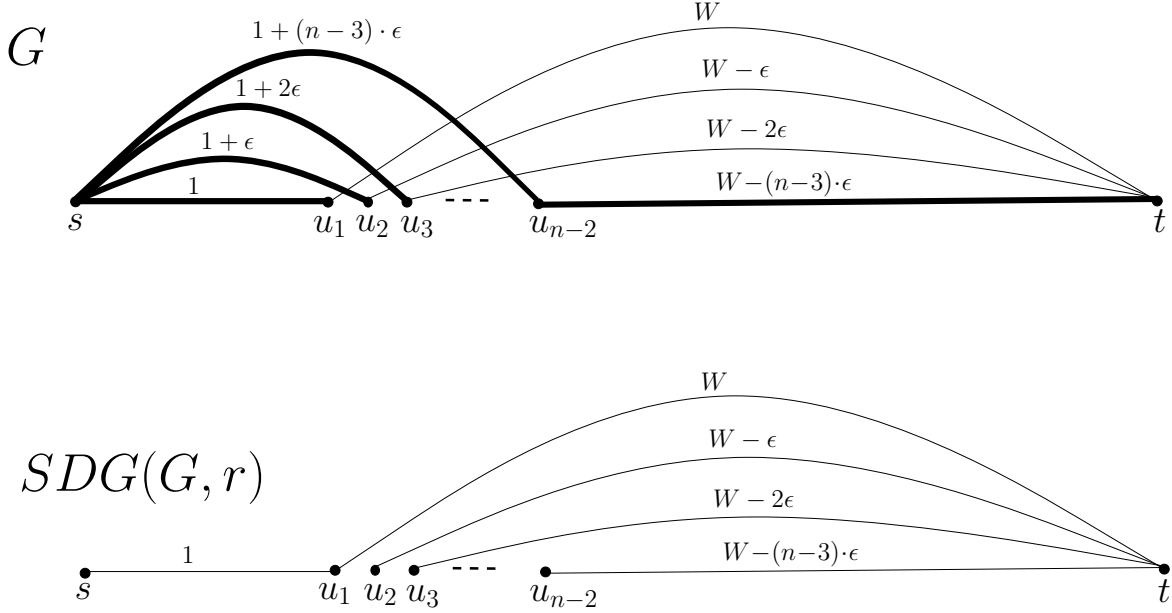}
\end{center}
\end{minipage}
\caption[]{ \label{1dim} \sf \footnotesize The graph $G$ that appears at the top of the figure is a 1-dimensional Euclidean  $n$-vertex graph
that contains $2n-4$ edges. The $n-1$ edges of the MST of $G$ are depicted by bold lines.
The graph $SDG(G,r)$ that appears at the bottom of the figure
is the SDG corresponding to $G$ and $r$, where $r(s)=1$, $r(u_1) = r(u_2) = \ldots = r(u_{n-2}) = r(t) = W$.}
\end{center}
\end{figure*}
Observe that the MST of $G$ contains all $n-2$ edges $\{(s,u_i) ~\vert~ u_i \in U\}$
that are incident to $s$ in $G$ and another edge $(u_{n-2},t)$ that connects $t$ with its closest neighbor $u_{n-2}$ in $G$.
Thus, $w(MST(G)) \approx (n-2) + W$. Consider the range assignment $r$ that assigns $r(x) = W$, for any point $x \in V \setminus \{s\}$,
and $r(s) = 1$.
The SDG corresponding to $G$ and $r$, denoted $SDG(G,r)$, contains all $n-2$ edges of $G$ that are incident to $t$ and another
edge $(s,u_1)$ that connects $s$ with its closest neighbor $u_1$ in $G$. 
In particular, $SDG(G,r)$ is a spanning tree for $G$ of weight roughly $(n-2)\cdot W+1$.
Thus, the weight-coefficient of $SDG(G,r)$ with respect to $G$ is roughly $\frac{(n-2) \cdot W + 1}{(n-2) + W} \approx n-2$,
assuming $W \gg n$ and $\epsilon \ll \frac{1}{n}$. In other words, we obtained a lower bound of $\Omega(n)$ on the weight-coefficient of SDGs for 1-dimensional Euclidean $n$-vertex graphs.

\subsection{The Weight is Exponentially Better than Other Parameters} \label{appb}
Our main result (Theorem \ref{final}) implies that the weight-coefficient of SDGs for arbitrary $n$-point metrics is $O(\log n)$.
In contrast, we showed (see Section 1.1) that the degree-coefficient of SDGs for $n$-point metrics can be as large as $\Omega(n)$. 
(The formal definition of degree-coefficient is given below.)
In this section we demonstrate that, similarly to the degree parameter, a lower bound of $\Omega(n)$ 
also applies to other basic parameters of spanning trees.


Let $T = (T,rt)$ be a (possibly weighted) rooted tree.
Before we proceed, we provide the formal definitions of some basic parameters of  trees:
\begin{itemize}
\item The \emph{maximum degree} (or shortly, \emph{degree}) of $T$  is the maximum number of edges that are incident
to some vertex in $T$. 
\item The \emph{radius} (respectively, \emph{depth}) 
of $T$ is the maximum weighted (resp., unweighted) distance between
$rt$ and some leaf in $T$. 
\item The \emph{diameter} (respectively, \emph{hop-diameter})
of $T$ is the maximum weighted (resp., unweighted) distance between some pair of vertices in $T$.
\item The \emph{sum of all pairwise distances} (or shortly, \emph{sum-pairwise}) of $T$ is defined as the 
sum of weighted distances between
all pairs of vertices in $T$.
\item 
The \emph{sum of all single-source distances} (or shortly, \emph{sum-single}) of $T$ 
is 
defined as the sum of weighted distances between $rt$ and all other vertices in $T$.
\end{itemize}
(The sum-pairwise and sum-single parameters also have unweighted versions.)

Let $G$ be a (possibly weighted) graph and let $G'$ be a connected spanning subgraph of $G$.
The \emph{degree-coefficient} of $G'$ with respect to $G$
is defined as the ratio between the degree of the minimum-degree 
spanning tree of $G'$ and the degree of the minimum-degree spanning tree of $G$.
In exactly the same way we can define the \emph{radius-coefficient}, \emph{depth-coefficient},
\emph{diameter-coefficient}, \emph{hop-diameter-coefficient}, \emph{sum-pairwise-coefficient},
 and \emph{sum-single-coefficient}.

In Section 1.1 we showed that the degree-coefficient of SDGs can be as large as $\Omega(n)$ for $n$-point metrics.
Next, we provide a simple example for which each of the other parameters defined above incurs the same lower bound of $\Omega(n)$. 

Consider the $n$-point metric $M = (V,\delta)$, where $V = \{v_1,v_2,\ldots,v_n\}$, and the
distance function $\delta$ satisfies that $\delta(v_1,v_2) = \delta(v_2,v_3) = \ldots = \delta(v_{n-1},v_n) = 1$,
and for all other pairs  of points $v_i, v_j \in V$, $\delta(v_i,v_j) = 2$.   
\\Let $SDG(M,r)$ be the SDG corresponding to $M$ and the range assignment $r \equiv 1$. It is easy
to see that $SDG(M,r)$ is the unweighted $n$-path $(v_1,v_2,\ldots,v_n)$. 

Let $T = (T,v_1)$ be the spanning tree of $M$ rooted at $v_1$ that consists of the $n-1$ edges $(v_1,v_2),(v_1,v_3),\ldots,(v_1,v_n)$.
Thus, $T$ is the $n$-star graph rooted at $v_1$ in which the weight of the edge $(v_1,v_2)$ is equal to $1$,
and all other edge weights are equal to 2.
Notice that $SDG(M,r)$ is already a spanning tree of $M$. Denote by $S = (S,v_1)$ the tree $SDG(M,r)$ rooted at $v_1$.
It is easy to see that:
\begin{itemize}
\item Both the radius and depth of $T$ are $O(1)$,
whereas the corresponding measures of $S$ are $\Omega(n)$. Thus, both the radius-coefficient
and depth-coefficient of $SDG(M,r)$ with respect to $M$ are $\Omega(n)$
\item Both the diameter and hop-diameter of $T$ are $O(1)$,
whereas the corresponding measures of $S$ are $\Omega(n)$. Thus, both the diameter-coefficient and
hop-diameter-coefficient of $SDG(M,r)$ with respect to $M$ are $\Omega(n)$.
\item The sum-pairwise of $T$ is $O(n^2)$
whereas the sum-pairwise of $S$ is $\Omega(n^3)$. Thus, the
sum-pairwise-coefficient of $SDG(M,r)$ with respect to $M$ is $\Omega(n)$.
\item The sum-single of $T$ is $O(n)$
whereas the sum-single of $S$ is $\Omega(n^2)$. Thus, the
sum-single-coefficient of $SDG(M,r)$ with respect to $M$ is $\Omega(n)$.
\end{itemize}
\section*{Acknowledgments}
The author thanks Michael Elkin for
helpful discussions.
	








\begin{thebibliography}{10}\setlength{\itemsep}{-1ex}\small

\bibitem{AACK10}
Abu-Affash, A.~K., Aschner, R.,  Carmi, P., Katz, M.~J.:
The MST of Symmetric Disk Graphs is Light.
In: Proc. of 12th SWAT, pp. 236--247 (2010).

\bibitem{ADDJS93}
Alth$\ddot{\mbox{o}}$fer, I., Das, G., Dobkin, D.~P., Joseph, D., Soares, J.:
On sparse spanners of weighted graphs.
Discrete \& Computational Geometry 9, 81--100 (1993).

\bibitem{BLRS02}
Blough, D.~M., Leoncini, M., Resta, G., Santi, P.: On the symmetric range assignment
problem in wireless ad hoc networks. In: Proc. of 
the IFIP 17th
World Computer Congress – TC1 Stream / 2nd IFIP International Conference on
Theoretical Computer Science (TCS), pp. 71–-82 (2002).

\bibitem{CMZ02}
Calinescu, G., Mandoiu, I.~I., Zelikovsky, A.: Symmetric connectivity with minimum
power consumption in radio networks. In: Proc. of 
the IFIP 17th
World Computer Congress – TC1 Stream / 2nd IFIP International Conference on
Theoretical Computer Science (TCS), pp. 119–-130 (2002).

\bibitem{CFKP07}
Caragiannis, I., Fishkin, A.~V.,  Kaklamanis, C., Papaioannou, E.: 
A tight bound for online colouring of disk graphs. 
Theor. Comput. Sci. 384(2-3), 152--160 (2007).

\bibitem{CDNS95}
Chandra, B., Das, G., Narasimhan, G., Soares, J.: New sparseness results on graph
spanners. Int. J. Comput. Geometry Appl. 5, 125-–144 (1995).

\bibitem{Chazelle00}
Chazelle, B.: A minimum spanning tree algorithm with inverse-Ackermann type complexity.
J. ACM 47(6), 1028–-1047 (2000).

\bibitem{CCJ90}
Clark, B.~N., Colbourn, C.~J., Johnson, D.~S.: Unit disk graphs.
Discrete Mathematics 86(1-3), 165--177 (1990).


\bibitem{CPS99}
Clementi, A.~E.~F., Penna, P., Silvestri, R.: Hardness results for the power range assignment
problem in packet radio networks. 
In: Hochbaum, D.~S., Jansen, K., Rolim,
J.~D.~P., Sinclair, A. (Eds.) RANDOM 1999 and APPROX 1999. LNCS, vol. 1671,
pp. 197–-208. Springer, Heidelberg (1999).


\bibitem{CS09}
Czumaj, A., Sohler, C.: 
Estimating the Weight of Metric Minimum Spanning Trees in Sublinear Time. 
SIAM J. Comput. 39(3), 904--922 (2009).

\bibitem{DPP06}
Damian, M., Pandit, S., Pemmaraju, S.~V.: Local approximation schemes for topology control. 
In: Proc. of 25th PODC, pp. 208-217 (2006).

\bibitem{DPP062}
Damian, M., Pandit, S., Pemmaraju, S.~V.: Distributed Spanner Construction in Doubling Metric Spaces. 
In: Proc. of 10th OPODIS, pp. 157--171 (2006).

\bibitem{DF94}
Das, S.~K., Ferragina, P.,: An $o(n)$ Work EREW Parallel Algorithm for Updating MST. 
In: Proc. of 2nd ESA, pp. 331--342 (1994).

\bibitem{Elkin06}
Elkin, M.: 
An Unconditional Lower Bound on the Time-Approximation Trade-off for the Distributed Minimum Spanning Tree Problem. 
SIAM J. Comput. 36(2), 433--456 (2006).


\bibitem{EJS01}
Erlebach, T., Jansen, K., Seidel, E.: 
Polynomial-time approximation schemes for geometric graphs. 
In Proc. of 12th SODA, pp. 671--679 (2001).

\bibitem{FFF04}
Fiala, J., Fishkin, A.~V., Fomin, F.~V.: 
On distance constrained labeling of disk graphs. 
Theor. Comput. Sci. 326(1-3), 261--292 (2004).


\bibitem{HK01}
Hlin\v{e}n\'{y}, P., Kratochv\'{i}l, J.: Representing graphs by disks and balls.
Discrete Mathematics 229(1-3), 101--124 (2001).

\bibitem{KKT95}
Karger, D.~R., Klein, P.~N., Tarjan, R.~E: A Randomized Linear-Time Algorithm to Find Minimum Spanning Trees. 
J. ACM 42(2), 321--328 (1995).

\bibitem{KKKP00}
Kirousis, L., Kranakis, E., Krizanc, D., Pelc, A.: Power consumption in packet radio
networks. Theoretical Computer Science 243(1-2), 289-–305 (2000).

\bibitem{KRY94}
Khuller, S., Raghavachari, B., Young, N.~E.: 
Low degree spanning trees of small weight. 
In: Proc. of 26th STOC, pp. 412--421 (1994).

\bibitem{KMPS04}
Kumar, V.~S.~A., Marathe, M.~V., Parthasarathy, S.,  Srinivasan, A.: 
End-to-end packet-scheduling in wireless ad-hoc networks. 
In: Proc. of 15 SODA, pp. 1021--1030 (2004).


\bibitem{vanL05}
van Leeuwen, E.~J.: Approximation Algorithms for Unit Disk Graphs. 
In: Proc. of 31st WG, pp. 351--361 (2005).

\bibitem{LL06}
van Leeuwen, E.~J., van Leeuwen, J.: On the Representation of Disk Graphs.
Technical report UU-CS-2006-037, 
Department of Information and Computing Sciences, 
Utrecht University (2006).

\bibitem{Li03}
Li, X.-Y.: Approximate MST for UDG Locally. 
In: Proc. of 9th COCOON, pp. 364--373 (2003).

\bibitem{LWWF04}
Li, X.-Y., Wang, Y., Wan, P.-J., Frieder, O.: 
Localized Low Weight Graph and Its Applications in Wireless Ad Hoc Networks. 
In: Proc. of 23rd INFOCOM, (2004).

\bibitem{LWS04}
Li, X.-Y., Wang, Y., Song, W.-Z.: 
Applications of k-Local MST for Topology Control and Broadcasting in Wireless Ad Hoc Networks. 
IEEE Trans. Parallel Distrib. Syst. 15(12), 1057--1069 (2004).

\bibitem{PR02}
Pettie, P., Ramachandran, V.: 
An optimal minimum spanning tree algorithm. 
J. ACM 49(1), 16--34 (2002).


\bibitem{Parkash99}
Prakash, R.: Unidirectional links prove costly in wireless ad hoc networks. 
In: Proc. of 3rd DIAL-M, pp. 15--22 (1999).

\bibitem{PR10}
Peleg, D., Roditty, L.: Localized spanner construction for ad hoc networks with variable transmission range.
ACM Transactions 
on Sensor Networks 7(3), Article 25, (2010).

\bibitem{PR102}
Peleg, D., Roditty, L.: Relaxed Spanners for Directed Disk Graphs. 
In: Proc. of 27th STACS, pp. 609--620 (2010).

\bibitem{Salowe91}
Salowe, J.~S.: Construction of Multidimensional Spanner Graphs, 
with Applications to Minimum Spanning Trees. 
In: Proc. of 7th SoCG, pp. 256--261 (1991).

\bibitem{TD06}
Thai, M.~T., Du, D.-Z.: Connected Dominating Sets in Disk Graphs with Bidirectional Links.
IEEE Communications Letter 10(3), 138--140 (2006). 

\bibitem{TTD08}
Thai, M.~T., Tiwari, R., Du, D.-Z.:
On Construction of Virtual Backbone in Wireless Ad Hoc Networks with Unidirectional Links. 
IEEE Trans. Mob. Comput. 7(9), 1098--1109 (2008).

\bibitem{TWLZD07}
Thai, M.~T., Wang, F., Liu, D., Zhu, S., Du, D.-Z.:
Connected Dominating Sets in Wireless Networks with Different Transmission Ranges. 
IEEE Trans. Mob. Comput. 6(7), 721--730 (2007).

\bibitem{Watten05}
Wattenhofer, R.: Algorithms for ad hoc and sensor networks. 
Computer Communications 28(13), 1498--1504 (2005).

\bibitem{ZSN02}
Zhou, H., Shenoy, N.~V., Nicholls, W.: 
Efficient minimum spanning tree construction without Delaunay triangulation. 
Inf. Process. Lett. 81(5), pp. 271--276 (2002).

\end{thebibliography}
\end{document}